\documentclass[preprint,12pt,3p]{elsarticle}

\usepackage{makeidx}
\usepackage{amsmath}
\usepackage{amsthm}
\usepackage{amssymb}
\usepackage{amsfonts}
\usepackage{algorithmic}
\usepackage{dsfont}
\usepackage{float}
\usepackage[plainpages=false,pdfpagelabels=true,colorlinks=true,citecolor=blue,hypertexnames=false]{hyperref}
\usepackage[ruled,vlined]{algorithm2e}
\usepackage{mathtools}
\usepackage{extarrows}
\usepackage{url}
\usepackage{balance}
\usepackage{multirow}
\usepackage{color}
\usepackage{diagbox}
\usepackage{booktabs}
\usepackage{makecell}
\usepackage[switch]{lineno}
\newcolumntype{V}{!{\vrule width 1pt}}

\newtheorem{theorem}{Theorem}[section]
\newtheorem{lemma}[theorem]{Lemma}
\newtheorem{corollary}[theorem]{Corollary}

\newtheorem{definition}[theorem]{Definition}

\newtheorem{example}[theorem]{Example}
\newtheorem{remark}[theorem]{Remark}
\newtheorem{problem}[theorem]{Problem}

\def\gcd{{\rm gcd}}

\def\diag{{\rm diag}}

\def\GL{{\rm GL}}

\def\M{{\mathbf{M}}}
\def\N{{\mathbf{N}}}
\def\W{{\mathbf{W}}}
\def\F{{\mathbf{F}}}
\def\U{{\mathbf{U}}}
\def\V{{\mathbf{V}}}
\def\P{{\mathbf{P}}}
\def\Q{{\mathbf{Q}}}
\def\B{{\mathbf{B}}}
\def\G{{\mathbf{G}}}
\def\S{{\mathbf{S}}}

\journal{Elsevier}

\begin{document}

\begin{frontmatter}

\title{Smith normal forms of bivariate polynomial matrices}

\author[sju]{Dong Lu}
\ead{donglu@swjtu.edu.cn}

\author[klmm,ucas]{Dingkang Wang}
\ead{dwang@mmrc.iss.ac.cn}

\author[hunu]{Fanghui Xiao}
\ead{xiaofanghui@hunnu.edu.cn}

\author[stu]{Xiaopeng Zheng\corref{cor}}
\ead{xiaopengzheng@stu.edu.cn}

\cortext[cor]{Corresponding author}

\address[sju]{School of Mathematics, Southwest Jiaotong University, Chengdu 610031, China}

\address[klmm]{State Key Laboratory of Mathematical Sciences, Academy of Mathematics and Systems Science, Chinese Academy of Sciences, Beijing 100190, China}

\address[ucas]{School of Mathematical Sciences, University of Chinese Academy of Sciences, Beijing 100049, China}

\address[hunu]{MOE-LCSM, School of Mathematics and Statistics, Hunan Normal University, Changsha 410081, China}

\address[stu]{College of Mathematics and Computer Science, Shantou University, Shantou 515821, China}

\begin{abstract}
 In 1978, Frost and Storey asserted that a bivariate polynomial matrix is equivalent to its Smith normal form if and only if the reduced minors of all orders generate the unit ideal. In this paper, we first demonstrate by constructing an example that for any given positive integer $s$ with $s \geq 2$, there exists a square bivariate polynomial matrix $\M$ with the degree of $\det(\M)$ in $y$ equal to $s$, for which the condition that reduced minors of all orders generate the unit ideal is not a sufficient condition for $\M$ to be equivalent to its Smith normal form. Subsequently, we prove that for any square bivariate polynomial matrix $\M$ where the degree of $\det(\M)$ in $y$ is at most $1$, Frost and Storey's assertion holds. Using the Quillen-Suslin theorem, we further extend our consideration of $\M$ to rank-deficient and non-square cases.
\end{abstract}


\end{frontmatter}


\section{Introduction}

 The study of equivalence between polynomial matrices and their Smith normal forms has long been a foundational topic in both algebra and applications, bridging algebraic theory with practical engineering problems (see \cite{Bose1995Multi,Brown1993Matrices} and the references therein). The Smith normal form, a canonical diagonal form obtained via elementary row and column transformations, encapsulates essential invariants of a polynomial matrix, such as determinantal divisors and invariant factors, making it a powerful tool for analyzing matrix properties. For univariate polynomial matrices, the theory is well-established: since the ring of univariate polynomials over a field is a principal ideal domain (PID), every such matrix is equivalent to its Smith form \cite{Kailath1980Linear,Rosenbrock1970State}. However, for multivariate polynomial matrices (in two or more variables), the problem remains open due to the non-PID structure of multivariate polynomial rings, leading to rich and challenging research directions.

 Early investigations concentrated on bivariate polynomial matrices. In 1978, Frost and Storey \cite{Frost1978Equ} asserted the following: a necessary and sufficient condition for bivariate polynomial matrices to be equivalent to their Smith normal forms is that their reduced minors (see Definition \ref{def-reduced-minors}) of all orders generate the unit ideal. But later, they constructed a counterexample in \cite{Frost1981} to show that the above assertion does not hold. It has been proven that the condition that reduced minors of all orders generate the unit ideal is only a necessary condition for the equivalence of a multivariate polynomial matrix to its Smith normal form. Thus, researchers have begun to explore under what structural conditions of the matrix this condition serves as a necessary and sufficient condition.

 To date, researchers have focused primarily on two major classes of multivariate polynomial matrices, and substantial results have been achieved concerning the problem of their equivalence to Smith normal forms. Let $\M\in K[x_1,x_2,\ldots,x_n]^{l\times l}$ with $\det(\M) = (x_1 - f(x_2,\ldots,x_n))^t$, where $f\in K[x_2,\ldots,x_n]$ and $t$ is a positive integer. When $t=1$, Lin et al. \cite{Lin2006On} proved that $\M$ is equivalent to its Smith normal form. Subsequently, the main result obtained for $t=1$ were extended in \cite{Li2017OnT,Li2022TheS,Lu2023NewR,Liu2024The} to the case where $t \geq 2$ . In particular, Liu et al. \cite{Liu2024The} proved that a necessary and sufficient condition for $\M$ to be equivalent to its Smith normal form is that the reduced minors of all orders generate the unit ideal. In this case, the Smith normal form of $\M$ takes the general form, namely: $\diag\left((x_1 - f)^{t_1},\ldots,(x_1 - f)^{t_l}\right)$, where $t_1,\ldots,t_l$ are nonnegative integers and satisfy $t_1 \leq \cdots \leq t_l$ and $\sum_{i=1}^{l} t_i = t$. Thus, they completely solved the problem of equivalence between multivariate polynomial matrices of this type and their Smith normal forms. Let $\M\in K[x_1,x_2,\ldots,x_n]^{l\times l}$ with $\det(\M) = g(x_1)$, where $g\in K[x_1] \setminus \{0\}$. When $n=2$ and $g$ is an irreducible polynomial in $K[x_1]$, Li et al. \cite{LiD2019Some} proved that $\M$ is equivalent to its Smith normal form by constructing a homomorphism from $K[x_1,x_2]$ to $(K[x_1]/\langle g \rangle)[x_2]$. Subsequently, the cases where $g$ is a power of an irreducible polynomial in $K[x_1]$ or $n \geq 2$ were investigated in \cite{Zheng2023New,Lu2023Equiva,Guan2025NewR}. Lu et al. \cite{Lu2024OnThe} lifted the restrictions of ``irreducibility'' and ``$n=2$'', and using localization techniques \cite{Yengui2015Con} and the Quillen-Suslin theorem \cite{Quillen1976Projective,Suslin1976Projective}, proved that a necessary and sufficient condition for $\M$ to be equivalent to its Smith normal form is that reduced minors of all orders generate the unit ideal. This fully resolved the equivalence problem for such multivariate polynomial matrices and their Smith normal forms.

 While the equivalence problem between the two aforementioned classes of multivariate polynomial matrices and their Smith normal forms has been fully resolved, it remains an open problem for various other classes of multivariate polynomial matrices.

 This paper focuses primarily on the problem of equivalence between bivariate polynomial matrices and their Smith normal forms. Let $\M\in K[x,y]^{l\times l}$, where $\det(\M) \neq 0$. Since $\det(\M) \in K[x,y] \setminus \{0\}$, we may without loss of generality write $\det(\M)$ as:
 \begin{equation*}
  \det(\M) = a_0(x) + a_1(x) \cdot y + \cdots + a_s(x) \cdot y^s,
 \end{equation*}
 where $a_0,a_1,\ldots,a_s\in K[x]$ with $a_s \neq 0$, and $s$ is a nonnegative integer. If $s=0$, then the equivalence problem between $\M$ and its Smith normal form has been resolved in \cite{Lu2024OnThe}. Now, we construct an example to illustrate that for any positive integer $s$ with $s \geq 2$, there exists a bivariate polynomial matrix $\M$ such that reduced minors of all orders generate the unit ideal, while $\M$ is not equivalent to its Smith normal form.

 \begin{example}\label{intro-example}
  Let
  \[ \M = \begin{pmatrix} -1+x-y+y^2 & y \\  0 &  1+x-y^{s-2} \end{pmatrix}\]
  be a bivariate polynomial matrix in $\mathbb{Q}[x,y]^{2\times 2}$, where $s$ is a positive integer with $s \geq 2$, and $\mathbb{Q}$ is the field of rational numbers.

  By calculation, we have $\det(\M) = (-1+x-y+y^2)(1+x-y^{s-2})$. It easy to check that $1\times 1$ reduced minors of $\M$ generate the unit ideal $\mathbb{Q}[x,y]$. Since $\M$ is a square matrix, $2\times 2$ reduced minors of $\M$ generate $\mathbb{Q}[x,y]$. Next, we prove that $\M$ is not equivalent to its Smith normal form.

  Let $\mathbf{A} = \begin{pmatrix} 1+y-y^2 & -y \\  0 & -1+y^{s-2} \end{pmatrix}$. Then we have $\M = x\cdot \mathbf{I}_2 - \mathbf{A}$. We construct the following two matrices:
  \[\M_1 = \begin{pmatrix} -1+x-y+y^2 & 1 \\  0 & 1+x-y^{s-2} \end{pmatrix} ~ \text{and} ~
  \mathbf{A}_1 = \begin{pmatrix} 1+y-y^2 & -1 \\  0 & -1+y^{s-2} \end{pmatrix},\]
  then $\M_1 = x\cdot \mathbf{I}_2 - \mathbf{A}_1$. Let
  \[\U = \begin{pmatrix} 1 & 0 \\ -(1+x-y^{s-2}) & 1 \end{pmatrix} ~ \text{and} ~
  \V = \begin{pmatrix} 0 & -1 \\  1 & -1+x-y+y^2 \end{pmatrix}.\]
  Then $\det(\U) = \det(\V) = 1$. This implies that $\U,\V$ are unimodular matrices. Moreover, we have
  \[\U\M_1\V = \begin{pmatrix} 1 &   \\    & (-1+x-y+y^2)(1+x-y^{s-2}) \end{pmatrix}.\]
  It follows that $\M_1$ is equivalent to its Smith normal form. By the fact that $x\cdot \mathbf{I}_2 - \mathbf{A}$ and $x\cdot \mathbf{I}_2 - \mathbf{A}_1$ are equivalent over $K[x,y]$ if and only if $\mathbf{A}$ and $\mathbf{A}_1$ are similar over $K[y]$ \cite{Brown1993Matrices}. We assume that there exists a unimodular matrix $\mathbf{Q}=(q_{ij})_{2\times 2} \in K[y]^{2\times 2}$ such that $\mathbf{Q}^{-1}\mathbf{A}\mathbf{Q} = \mathbf{A}_1$. Then we get
  \[\begin{cases}
     q_{11}q_{22} - q_{21}q_{12} = 1, \\
     (1+y-y^2)q_{11} - yq_{21} = (1+y-y^2)q_{11}, \\
     (1+y-y^2)q_{12} - yq_{22} = -q_{11} + (-1+y^{s-2})q_{12},  \\
     (-1+y^{s-2})q_{21} = (1+y-y^2)q_{21}, \\
     (-1+y^{s-2})q_{22} = -q_{21} + (-1+y^{s-2})q_{22}.\end{cases}\]
  It is easy to check that $q_{21} = 0$, $q_{11}q_{22} = 1$, and $q_{11} - yq_{22} = (-2-y+y^2+y^{s-2})q_{12}$. It follows from $q_{11}q_{22} = 1$ that the degree of $q_{11} - yq_{22}$ in $y$ is $1$. However, the degree of $(-2-y+y^2+y^{s-2})q_{12}$ in $y$ is greater than $1$, which leads to a contradiction. Therefore, $\M$ is not equivalent to $\M_1$. This implies that $\M$ is not equivalent to its Smith normal form.
 \end{example}

 Example \ref{intro-example} reveals the fact that, in the general case, when the degree of $\det(\M)$ in $y$ is greater than $1$, the property that the reduced minors of all orders of $\M$ generate the unit ideal is not a sufficient condition for $\M$ to be equivalent to its Smith normal form. Based on the above fact, we shall consider the following problem.

 \begin{problem}\label{intro_problem}
  Let $\M\in K[x,y]^{l\times l}$ and $\det(\M) = f(x)y+g(x)$, where $f,g\in K[x]$ and $f \neq 0$. Is the necessary and sufficient condition for $\M$ to be equivalent to its Smith normal form that the reduced minors of all orders of $\M$ generate the unit ideal?
 \end{problem}

 The rest of the paper is organized as follows. In Section \ref{sec1}, we recall some terminology and preliminary results needed for this paper. In Section \ref{sec3}, we solve Problem \ref{intro_problem} and extend the square matrix case to rank-deficient and non-square matrix cases. Some concluding remarks are provided in Section \ref{sec5}.

\section{Preliminaries}\label{sec1}

 Let $K$ be a field, $K[x,y]$ be the bivariate polynomial ring in variables $x,y$ with coefficients in $K$, and $K[x,y]^{l\times m}$ be the set of $l\times m$ matrices with entries in $K[x,y]$, where $l,m$ are two positive integers. Throughout this paper, we assume without loss of generality that $l\leq m$.

 The arguments $(x,y)$ and $(x)$ are omitted whenever their omission does not cause confusion; for example, we denote $f(x,y)$ and $p(x)$ by $f$ and $p$ for simplicity, respectively. Let $f_1, \ldots, f_l\in K[x,y]$; we use $\langle f_1, \ldots, f_l \rangle$ and $\diag(f_1, \ldots, f_l)$ to represent the ideal in $K[x,y]$ generated by $f_1, \ldots, f_l$ and the diagonal matrix whose diagonal elements are $f_1, \ldots, f_l$, respectively.

 For convenience, we use bold letters to denote polynomial matrices. Let $\mathbf{M} \in K[x,y]^{l \times m}$. For $i=1,\ldots,l$, we use $I_i(\mathbf{M})$ and $d_i(\mathbf{M})$ to denote the ideal generated by all the $i \times i$ minors of $\mathbf{M}$ and the greatest common divisor of all the $i \times i$ minors of $\mathbf{M}$, respectively. Here, we make the convention that $d_0(\mathbf{M}) \equiv 1$.

 We first introduce the concept of the Smith normal form of bivariate polynomial matrices.

 \begin{definition}\label{def}
  Let $\mathbf{M} \in K[x,y]^{l \times m}$ with rank $r$, and $f_i$ be a polynomial in $K[x,y]$ defined as follows:
  \[f_i= \begin{cases}\frac{d_i(\mathbf{M})}{d_{i-1}(\mathbf{M})}, & 1 \leq i \leq r, \\ 0, & r<i \leq l,\end{cases}\]
  Then the Smith normal form of $\mathbf{M}$ is given by
  \[\begin{pmatrix}
   \diag(f_1, \ldots, f_r) & \mathbf{0}_{r \times(m-r)} \\
		\mathbf{0}_{(l-r) \times r} & \mathbf{0}_{(l-r) \times(m-r)}
  \end{pmatrix}.\]
 \end{definition}

 In Definition \ref{def}, Li et al. \cite{Li2025Smith} used localization and $p$-index techniques to prove that $f_1\mid f_2\mid \cdots \mid f_r$. This indicates that the above definition is a natural generalization of the concept of the Smith normal form of univariate polynomial matrices.

\begin{definition}\label{def-uni}
 Let $\mathbf{U} \in K[x,y]^{l \times l}$. Then $\mathbf{U}$ is said to be unimodular if $\det(\mathbf{U})$ is a unit in $K[x,y]$. The set of all $l\times l$ unimodular matrices over $K[x,y]$ is denoted by $\GL_l(K[x,y])$.
\end{definition}

 With the set of unimodular matrices defined, we can now specify what it means for two polynomial matrices to be equivalent.

\begin{definition}
 Let $\mathbf{M}, \mathbf{N} \in K[x,y]^{l \times m}$. We say that $\mathbf{M}$ is equivalent to $\mathbf{N}$ over $K[x,y]$ if there exist $\mathbf{U} \in \GL_l(K[x,y])$ and $\mathbf{V} \in \GL_m(K[x,y])$ such that $\mathbf{N} =\mathbf{U} \mathbf{M} \mathbf{V}$.
\end{definition}

 In what follows, we write $\mathbf{M} \sim \mathbf{N}$ to denote the equivalence of $\mathbf{M}$ and $\mathbf{N}$ over $K[x,y]$.

\begin{definition}[\cite{Lin1988}]\label{def-reduced-minors}
 Let $\mathbf{M} \in K[x,y]^{l \times m}$ with rank $r$, where $1 \leq r \leq l$. For any given integer $i$ with $1 \leq i \leq r$, let $a_1, \ldots, a_\beta$ be all the $i \times i$ minors of $\mathbf{M}$, where $\beta=\binom{l}{i}\binom{m}{i}$. Extracting $d_i(\mathbf{M})$ from $a_1, \ldots, a_\beta$ yields
 \[a_j=d_i(\mathbf{M}) \cdot b_j, ~~ j=1,\ldots,\beta.\]
 Then, $b_1, \ldots, b_\beta$ are called the $i \times i$ reduced minors of $\mathbf{M}$. For convenience, we use $J_i(\mathbf{M})$ to denote the ideal generated by $b_1, \ldots, b_\beta$. In addition, set $J_0(\mathbf{M}) \equiv K[x,y]$.
\end{definition}

 A crucial observation is that $d_i(\mathbf{M})$, $I_i(\M)$ and $J_i(\mathbf{M})$ are invariant under matrix equivalence, as shown in the following lemma, whose proof follows straightforwardly from the Cauchy-Binet formula \cite{Strang2010Linear}.

\begin{lemma}\label{lem0}
 Let $\M,\N \in K[x,y]^{l\times m}$. If $\M \sim \N$, then $d_i(\M) = d_i(\N)$, $I_i(\M) = I_i(\N)$ and $J_i(\M) = J_i(\N)$, where $i = 1,\ldots,l$.
\end{lemma}

\begin{lemma}[\cite{Lu2024OnThe}]\label{lem3}
 Let $\M,\F,\G \in K[x,y]^{l\times l}$ satisfy $\M = \G\F$, and $\gcd(\det(\G),\det(\F))=1$. If $J_{i}(\M) = K[x,y]$, then $J_i(\G) = J_i(\F) = K[x,y]$, where $i = 1,\ldots,l$.
\end{lemma}

 Building on the concept of reduced minors, the following lemma establishes a criterion for a square matrix to be equivalent to its Smith normal form, linking the ideal $J_i(\M)$ to the equivalence property.

\begin{lemma}[\cite{Lu2024OnThe}]\label{lem4}
 Let $\M \in K[x,y]^{l\times l}$ with $\det(\M) \in K[x]\setminus \{0\}$. Then $\M$ is equivalent to it Smith normal form if and only if $J_i(\M) = K[x,y]$ for $i = 1,\ldots,l$.
\end{lemma}

 We now introduce the notion of zero left prime matrices, which play a role in understanding their relation to unimodular transformations and the decomposability of matrices.

\begin{definition}[\cite{Youla1979Notes}]
 Let $\mathbf{M} \in K[x,y]^{l \times m}$ be of full rank. Then $\mathbf{M}$ is said to be zero left prime (ZLP) if all the $l\times l$ minors of $\mathbf{M}$ generate $K[x,y]$. Similarly, $\mathbf{M} \in K[x,y]^{m \times l}$ can be defined as a zero right prime (ZRP) matrix.
\end{definition}

\begin{theorem}[Quillen-Suslin Theorem, \cite{Quillen1976Projective, Suslin1976Projective}]\label{QS}
 Let $\mathbf{M} \in K[x,y]^{l \times m}$ be a ZLP matrix. Then there exists $\mathbf{V}\in \GL_m(K[x,y])$ such that $\mathbf{M}\mathbf{V} = (\mathbf{I}_l, \mathbf{0}_{l\times (m-l)})$, where $\mathbf{I}_l$ is the $l\times l$ identity matrix.
\end{theorem}

\begin{theorem}[\cite{Wang2004}]\label{lin-bose}
 Let $\mathbf{M} \in K[x,y]^{l \times m}$ with rank $r$, where $1\leq r \leq l$. If $J_r(\mathbf{M}) = K[x,y]$, then there exist $\mathbf{G} \in K[x,y]^{l \times r}$ and $\mathbf{F} \in K[x,y]^{r \times m}$ such that $\mathbf{M}=\mathbf{G} \mathbf{F}$, where $d_r(\mathbf{G}) = d_r(\mathbf{M})$ and $\mathbf{F}$ is a ZLP matrix.
\end{theorem}

\begin{lemma}[Primitive Factorization Theorem, \cite{Guiver1982Poly}]\label{lemma-PFT}
 Let $\M\in K[x,y]^{l\times m}$ be of full row rank, and $h\in K[x]$ be a divisor of $d_l(\M)$. Then there exist $\G\in K[x,y]^{l\times l}$ and $\F\in K[x,y]^{l\times m}$ such that $\M = \G\F$ and $\det(\G) = h$.
\end{lemma}

\section{Matrix equivalence theory}\label{sec3}

 Let $\M\in K[x,y]^{l\times l}$ and $\det(\M) = fy+g$, where $f,g\in K[x]$ and $f \neq 0$. In this section, we first consider the case where $\gcd(f,g)=1$, and prove that $\M$ must be equivalent to its Smith normal form. Second, we deal with the case where $\gcd(f,g)$ is a nontrivial polynomial in $K[x]$, and solve Problem \ref{intro_problem} by means of mathematical induction. Subsequently, we extend the square matrix case to rank-deficient and non-square cases using the Quillen-Suslin theorem, thereby obtaining a more general conclusion. Finally, we illustrate the validity of the equivalence theory through an example.

\begin{lemma}\label{lem1}
 Let $h\in K[x,y]\setminus \{0\}$, and $p,q\in K[x]$ be such that $p \neq 0$ and $\gcd(p,q)=1$. Then $(py+q) \mid h$ if and only if $h(x,-\frac{q}{p}) = 0$.
\end{lemma}

\begin{proof}
 The necessity is straightforward to establish. It remains to prove the sufficiency. By the polynomial pseudo-division over $K[x,y]$, there exist $h_1\in K[x,y]$, $\gamma\in K[x]$ and some nonnegative integer $s\in \mathbb{N}$ such that
 \begin{equation}\label{equ-lem1-2}
  p(x)^s\cdot h(x,y) = h_1(x,y)\cdot (p(x)y+q(x))+\gamma(x).
 \end{equation}
 Replacing $y$ with $-\frac{q}{p}$ in Equation \eqref{equ-lem1-2}, we obtain $\gamma(x) =0$. This implies that $(py+q) \mid p^sh$. It follows from $\gcd(p,q) = 1$ that $py+q$ is an irreducible polynomial in $K[x,y]$. Since $p\in K[x]$, we get $(py+q) \mid h$. The proof is completed.
\end{proof}

\begin{lemma}\label{lem2}
 Let $\M \in K[x,y]^{l \times l}$ with $\det(\M) = py+q$, where $p,q\in K[x]$ satisfy $p \neq 0$ and $\gcd(p,q)=1$. Then there exist $\V\in \GL_l(K[x])$ and $\G\in \GL_l(K[x,y])$ such that $\M = \V \cdot \diag(1,\ldots,1,py+q) \cdot \G$.
\end{lemma}

\begin{proof}
 Let $\N = \M(x,-\frac{q(x)}{p(x)})$. Then $\N \in K(x)^{l\times l}$ and $\det(\N) = 0$. Let $s \in \mathbb{N}$ be a sufficiently large positive integer such that $p^s\cdot \N \in K[x]^{l \times l}$. As $K[x]$ is a Euclidean domain, there exists $\U\in \GL_l(K[x])$ such that
 \begin{equation}\label{equ-lem2-1}
  \U(p^s\cdot \N) = \begin{pmatrix}
		\N' \\
		\mathbf{0}_{1\times l}
	\end{pmatrix},
 \end{equation}
 where $\N' \in K[x]^{(l-1) \times l}$. Since $\U(p^s\cdot \N) = p^s\cdot (\U\N)$, it follows from Equation \eqref{equ-lem2-1} that the last row of $\U\N$ is a zero vector. Let $\M' = \U\M$. Then
 \begin{equation*}\label{equ-lem2-2}
  \M'(x,-\frac{q(x)}{p(x)}) = \U\N.
 \end{equation*}
 Thus, the last row of $\M'(x,-\frac{q(x)}{p(x)})$ is a zero vector. By Lemma \ref{lem1}, all entries in the last row of $\M'$ are divisible by $py+q$. Hence, there exists $\G \in K[x,y]^{l \times l}$ such that
 \begin{equation}\label{equ-lem2-3}
  \U\M = \diag(1,\ldots,1,py+q)\cdot \G.
 \end{equation}
 Since $\det(\M) = py+q$ and $\U$ is a unimodular matrix, it follows from Equation \eqref{equ-lem2-3} that $\G \in \GL_l(K[x,y])$. Let $\V = \U^{-1}$. Then
 \begin{equation*}\label{equ-lem2-4}
  \M = \V \cdot \diag(1,\ldots,1,py+q) \cdot \G.
 \end{equation*}
 The proof is completed.
\end{proof}

\begin{remark}
 If $q=0$ in Lemmas \ref{lem1} and \ref{lem2}, it follows from $\gcd(p,q) = 1$ that $p=1$. In this case, $s$ in both Equations \eqref{equ-lem1-2} and \eqref{equ-lem2-1} is $0$, and $\N$ in Equation \eqref{equ-lem2-1} is a polynomial matrix in $K[x]^{l\times l}$.
\end{remark}

 It is easy to verify that the Smith normal form of $\M$ in Lemma \ref{lem2} is $ \diag(1,\ldots,1,py+q)$. Lemma \ref{lem2} implies that $\M$ is equivalent to its Smith normal form.

\begin{lemma}\label{main-lemma-1}
 Let $f,v_1,v_2\in K[x]$ with $\gcd(f,v_1,v_2) = 1$, where $f \neq 0$. Then there exists $h\in K[x]$ such that $\gcd(f,v_1+hv_2) = 1$.
\end{lemma}

\begin{proof}
 There are two cases. If $f\in K \setminus \{0\}$, then the conclusion obviously holds for any polynomial $h\in K[x]$. If $f\in K[x] \setminus K$, then we may assume without loss of generality that $f = \alpha p_1^{r_1}\cdots p_k^{r_k}$ is an irreducible factorization of $f$, where $p_1,\ldots,p_k\in K[x]$ are irreducible factors that are pairwise  coprime, $r_1,\ldots,r_k$ are positive integers and $\alpha\in K \setminus \{0\}$. For each integer $i$ with $1\leq i \leq k$, it is easy to see that $p_i$ cannot divide both $v_1$ and $v_2$; otherwise, it would contradict the fact that $\gcd(f,v_1,v_2) = 1$. This implies that there exits some integer $\delta_i\in \{0,1\}$ such that
 \begin{equation}\label{equ-mian-lemma-1}
  v_1 + \delta_i v_2 \neq 0  ~ (\text{mod} ~ p_i).
 \end{equation}
 By the Chinese Remainder Theorem, there exists $h\in K[x]$ such that
 \begin{equation}\label{equ-mian-lemma-2}
  h \equiv \delta_i ~ (\text{mod} ~ p_i).
 \end{equation}
 Combining Equations \eqref{equ-mian-lemma-1} and \eqref{equ-mian-lemma-2}, we have
 \begin{equation}\label{equ-mian-lemma-3}
  v_1 + h v_2 \neq 0  ~ (\text{mod} ~ p_i) ~ \text{for each} ~ i=1,\ldots,k.
 \end{equation}
 It follows from Equation \eqref{equ-mian-lemma-3} that
 \[ \gcd(f,v_1+hv_2) = 1.\]
 The proof is completed.
\end{proof}

\begin{lemma}\label{lemW}
 Let
 \[\W= \begin{pmatrix}
		f_1 &       0     &   \cdots   &    0      & v_1(py+q)    \\
		 0  &      f_2    &    \ddots  &  \vdots   & v_2(py+q)    \\
     \vdots &    \ddots   &    \ddots  &    0      & \vdots       \\
	 \vdots &             &   \ddots   &  f_{l-1}  & v_{l-1}(py+q)\\
		0   &   \cdots    &  \cdots    &     0     & v_{l}(py+q)
 \end{pmatrix}\in K[x,y]^{l\times l},\]
 where $f_{1},f_2,\ldots,f_{l-1}\in K[x] \setminus \{0\}$ satisfy $f_1\mid f_2 \mid \cdots \mid f_{l-1}$, and $p,q,v_{1},\ldots,v_l\in K[x]$. If $I_1(\W) = K[x,y]$, then $\langle f_1, py+q \rangle = K[x,y]$ and $\gcd(f_1,v_1,\ldots,v_l) = 1$.
\end{lemma}

\begin{proof}
 As $I_1(\W) = \langle f_1, \ldots, f_{l-1}, v_1(py+q),\ldots,v_l(py+q) \rangle$, it follows from $f_1\mid f_2 \mid \cdots \mid f_{l-1}$ that
 \[ I_1(\W) \subseteq \langle f_1, py+q \rangle ~ \text{and} ~ I_1(\W) \subseteq \langle f_1, v_1,\ldots,v_l \rangle.\]
 By the fact that $I_1(\W) = K[x,y]$, we have $\langle f_1, py+q \rangle = \langle f_1, v_1,\ldots,v_l \rangle = K[x,y]$. Obviously, $\gcd(f_1,v_1,\ldots,v_l) = 1$. The proof is completed.
\end{proof}

\begin{lemma}\label{lem7}
 Let $\M \in K[x,y]^{l\times l}$ with $\det(\M) = fy+g$, where $f,g\in K[x]$ and $f \neq 0$. If $J_i(\M) = K[x,y]$ for $i=1,\ldots,l$, then $\M \sim (\B_1,(py+q)\mathbf{b})$, where $\mathbf{B}_1 \in K[x]^{l \times (l-1)}$, $\mathbf{b} \in K[x]^{l\times 1}$ and $p,q\in K[x]$ satisfy $p \neq 0$ and $\gcd(p,q) = 1$.
\end{lemma}

\begin{proof}
 Let $h = \gcd(f,g)$. Then there exist $p,q\in K[x]$ such that
 \[ f = hp ~\text{and} ~ g = hq,\]
 where $p \neq 0$ and $\gcd(p,q)=1$. Since $\det(\M) = h(py+q)$ and $h\in K[x]$, by the Primitive Factorization Theorem (Lemma \ref{lemma-PFT}) there exist $\G, \F \in K[x,y]^{l \times l}$ such that
 \begin{equation}\label{equ-lem7-1}
  \M = \G\F ~ \text{and} ~ \det(\G) = h.
 \end{equation}
 Since $\det(\M) = \det(\G)\cdot \det(\F)$, it follows from Equation \eqref{equ-lem7-1} that $\det(\F) = py+q$. According to Lemma \ref{lem3}, we have $J_i(\G) = K[x,y]$ for all $i = 1, \ldots, l$. Based on Lemma \ref{lem4}, $\G$ is equivalent to its Smith normal form $\S_{\G}$, where $\S_{\G}\in K[x]^{l\times l}$. That is, there exist $\P, \Q \in \GL_l(K[x,y])$ such that
 \begin{equation}\label{equ-lem7-2}
  \G = \P \cdot \S_{\G} \cdot \Q.
 \end{equation}
 Combining Equations \eqref{equ-lem7-1} and \eqref{equ-lem7-2}, we obtain
 \begin{equation}\label{equ-lem7-3}
  \M = \P \cdot \S_{\G} \cdot \Q \cdot \F.
 \end{equation}
 Let $\F' = \Q \cdot \F$. By the fact that $\Q$ is unimodular, we get $\det(\F') = \det(\F) = py+q$. By Lemma \ref{lem2}, there exist $\V \in \GL_l(K[x])$ and $\F_1 \in \GL_l(K[x,y])$ such that
 \begin{equation}\label{equ-lem7-4}
  \F' = \V \cdot \diag(1, \ldots, 1, py + q) \cdot \F_1.
 \end{equation}
 It follows from Equations \eqref{equ-lem7-3} and \eqref{equ-lem7-4} that
 \begin{equation}\label{equ-lem7-5}
  \M = \P \cdot \S_{\G} \cdot \V \cdot \diag(1, \ldots, 1, py + q) \cdot \F_1.
 \end{equation}
 Obviously, $\S_{\G} \cdot \V\in K[x]^{l\times l}$. Let $\B_1\in K[x]^{l \times (l-1)}$ be the matrix formed by the first $l-1$ columns of $\S_{\G} \cdot \V$, and let $\mathbf{b} \in K[x]^{l\times 1}$ be the last column of $\S_{\G} \cdot \V$. Then
 \begin{equation*}\label{equ-lem7-6}
  \M = \P \cdot (\B_1,(py+q)\mathbf{b}) \cdot \F_1.
 \end{equation*}
 Since $\P,\F_1$ are unimodular matrices, we have
 \begin{equation*}\label{equ-lem7-7}
  \M \sim (\B_1,(py+q)\mathbf{b}).
 \end{equation*}
 The proof is completed.
\end{proof}
 	
\begin{lemma} \label{lem8}
 Let $\mathbf{B} = (\mathbf{B}_1, (py + q)\mathbf{b}) \in K[x,y]^{l \times l}$ with $\det(\B) \neq 0$, where $\mathbf{B}_1 \in K[x]^{l \times (l - 1)}$, $\mathbf{b} \in K[x]^{l \times 1}$, and $p, q \in K[x]$. If $I_1(\mathbf{B}) = K[x,y]$, then
 \[\mathbf{B} \sim
	\begin{pmatrix}
		1 & \mathbf{0}_{1\times (l-1)} \\
		\mathbf{0}_{(l-1)\times 1} & \M_1
	\end{pmatrix},\]
 where $\M_1 \in K[x,y]^{(l-1) \times (l-1)}$.
\end{lemma}

\begin{proof}
 Since $\det(\B) \neq 0$, $\B_1 \in K[x]^{l \times (l-1)}$ is a full column rank matrix. By the fact that $K[x]$ is a Euclidean domain, there exist $\P \in \mathrm{GL}_l(K[x])$ and $\Q \in \mathrm{GL}_{l-1}(K[x])$ such that
 \begin{equation*}\label{equ-lem8-1}
  \P\B_1\Q =
	\begin{pmatrix}
		f_1  &  0     &   \cdots  & 0       \\
		 0   & f_2    &   \ddots  & \vdots  \\
	  \vdots &\ddots  &   \ddots  & 0       \\
	  \vdots &        &  \ddots   & f_{l-1} \\
		0    & \cdots &  \cdots   &  0
	\end{pmatrix},
 \end{equation*}
 where $f_1,f_2, \ldots, f_{l-1} \in K[x] \setminus \{0\}$ satisfy $f_1\mid f_2 \mid \cdots \mid f_{l-1}$. Let $\W = \P\B\begin{pmatrix}\Q & \\ & 1 \end{pmatrix}$. Then
 \begin{equation*}\label{equ-lem8-2}
  \W= \begin{pmatrix}
		f_1 &       0     &   \cdots   &    0      & v_1(py+q)    \\
		 0  &      f_2    &    \ddots  &  \vdots   & v_2(py+q)    \\
     \vdots &    \ddots   &    \ddots  &    0      & \vdots       \\
	 \vdots &             &   \ddots   &  f_{l-1}  & v_{l-1}(py+q)\\
		0   &   \cdots    &  \cdots    &     0     & v_{l}(py+q)
 \end{pmatrix},
 \end{equation*}
 where $(v_1, \ldots, v_l)^{\rm T} = \P \cdot \mathbf{b}$. By Lemma \ref{lem0}, we have $I_1(\W) = I_1(\B) = K[x,y]$. According to Lemma \ref{lemW}, we obtain
 \begin{equation}\label{equ-lem8-3}
  \langle f_1, py+q \rangle = K[x,y] ~ \text{and} ~ \gcd(f_1,v_1,\ldots,v_l) = 1.
 \end{equation}
 Let $d = \gcd(v_2,\ldots,v_l)$. Then there exist $c_2,\ldots,c_n\in K[x]$ such that
 \begin{equation*}\label{equ-lem8-4}
  d = c_2v_2+\cdots+c_lv_l.
 \end{equation*}
 It follows from $\gcd(f_1,v_1,v_2,\ldots,v_n) = 1$ that $\gcd(f_1,v_1,d) = 1$. By Lemma \ref{main-lemma-1}, there exists $h\in K[x]$ such that $\gcd(f_1,v_1+hd) = 1$. As $f_1, v_1+hd \in K[x] \subset K[x,y]$, we have
 \begin{equation}\label{equ-lem8-5}
  \langle f_1,v_1+hd \rangle = K[x,y].
 \end{equation}
 Let
 \[\U = \begin{pmatrix}
		1 & hc_2 & \cdots &hc_{l-1} &  hc_l \\
		& 1   &        &  & &   \\
		&     & \ddots &  & &   \\
		&     &        & 1 & &   \\
		&     &        &   &1
	\end{pmatrix}
 ~ \text{and} ~~ \W' = \U\W.\]
 Then
 \begin{equation*}\label{equ-lem8-6}
  \W' =\begin{pmatrix}
		f_1 & hc_2 f_2 & \cdots & hc_{l-1} f_{l-1} & (v_1+hd)(py+q) \\
		& f_2     &        &                   & v_2(py + q) \\
		&            & \ddots &                   & \vdots \\
		&            &        & f_{l-1}        & v_{l-1}(py + q) \\
		&            &        &                   & v_l(py + q)
	\end{pmatrix}.
 \end{equation*}
 Combining Equations \eqref{equ-lem8-3} and \eqref{equ-lem8-5}, we deduce that
 \begin{equation*}\label{equ-lem8-7}
  \langle f_1, (v_1 + hd)(py + q) \rangle = K[x, y].
 \end{equation*}
 It follows that the first row of $\W'$ is a ZLP vector. By the Quillen-Suslin Theorem (Lemma \ref{QS}), there exists $\V \in \GL_l(K[x,y])$ such that
 \begin{equation*}\label{equ-lem8-8}
  \W'\V = 	\begin{pmatrix}
    	1 & \mathbf{0}_{1\times (l-1)} \\
    	\mathbf{b}_1 & \M_1
    \end{pmatrix},
 \end{equation*}
 where $\mathbf{b}_1\in K[x,y]^{(l-1) \times 1}$ and $\M_1\in K[x,y]^{(l-1)\times (l-1)}$. Furthermore, by performing elementary row operations on $\W'\V$, we get
 \begin{equation*}\label{equ-lem8-9}
  \W'\V \sim \begin{pmatrix}
    	1 & \mathbf{0}_{1\times (l-1)} \\
    	\mathbf{0}_{(l-1)\times 1} & \M_1
    \end{pmatrix}.
 \end{equation*}
 This implies that
 \begin{equation*}\label{equ-lem8-10}
  \mathbf{B} \sim
	\begin{pmatrix} 1 &  \\  & \M_1
	\end{pmatrix}.
 \end{equation*}
 The proof is completed.
\end{proof}

\begin{lemma}\label{lem9}
 Let $\M = \begin{pmatrix}  1 &  \\  & \M_1 \end{pmatrix}\in K[x,y]^{l\times l}$, where $\M_1 \in K[x,y]^{(l-1)\times (l-1)}$. Then $J_i(\M) = J_{i-1}(\M_1)$ for $i = 2,\ldots,l$.
\end{lemma}

\begin{proof}
 For any given integer $i$ with $2 \leq i \leq l$, it follows from the special form of $\M$ that any $i\times i$ minor of $\M$ is either an $i\times i$ minor of $\M_1$ or arises from an $(i-1)\times (i-1)$ minor of $\M_1$. This implies that
 \begin{equation*}\label{equ-lem9-1}
  I_i(\M) = I_i(\M_1) + I_{i-1}(\M_1).
 \end{equation*}
 Since $I_{i}(\M_1) \subseteq I_{i-1}(\M_1)$, we have
 \begin{equation}\label{equ-lem9-2}
  I_i(\M) = I_{i-1}(\M_1).
 \end{equation}
 It follows from Equation \eqref{equ-lem9-2} that $d_i(\M) = d_{i-1}(\M_1)$. By the fact that
 \begin{equation*}\label{equ-lem9-4}
  I_i(\M) = d_i(\M)\cdot J_i(\M) ~ \text{and} ~ I_{i-1}(\M_1) = d_{i-1}(\M_1) \cdot J_{i-1}(\M_1),
 \end{equation*}
 we obtain $J_i(\M) = J_{i-1}(\M_1)$. The proof is completed.
\end{proof}

\begin{theorem}\label{main-Theorem}
 Let $\M \in K[x,y]^{l\times l}$ with $\det(\M) = fy+g$, where $f,g\in K[x]$ and $f \neq 0$. Then $\M$ is equivalent to its Smith normal form if and only if $J_i(\M) = K[x,y]$ for $i = 1,\ldots,l$.
\end{theorem}

\begin{proof}
 The necessity is obvious from Lemma \ref{lem0}. It suffices to prove the sufficiency. We proceed by induction on $l$. The statement is clearly true when $l = 1$. Assume that the statement holds for all $l < k$.

 For $l = k$, we extract $d_1(\M)$ from $\M$ and obtain
 \begin{equation}\label{equ-maintheorem-1}
  \M = d_1(\M) \cdot \M_0.
 \end{equation}
 Obviously, $(d_1(\M))^k \mid \gcd(f,g)$, $d_1(\M_0) = 1$ and $\det(\M_0) \neq 0$. Moreover, it follows from $J_i(\M) = K[x,y]$ and Equation \eqref{equ-maintheorem-1} that $J_i(\M_0) = K[x,y]$, where $i=1,\ldots,k$. By Lemma \ref{lem7}, we have
 \begin{equation}\label{equ-maintheorem-2}
  \M_0 \sim (\B_1, (py + q)\mathbf{b}),
 \end{equation}
 where $\B_1 \in K[x,y]^{k \times (k - 1)}$, $\mathbf{b} \in K[x,y]^{k\times 1}$, and $p, q \in K[x]$ satisfy $p \neq 0$ and $\gcd(p,q)=1$. Let $\B = (\B_1, (py + q)\mathbf{b})$.  Since $\M_0 \sim \B$, we have
 \begin{equation*}\label{equ-maintheorem-3}
  \det(\B) \neq 0, ~ d_1(\B) = d_1(\M_0) = 1 ~ \text{and} ~ J_i(\B) =  J_i(\M_0) = K[x,y],
 \end{equation*}
 where $i=1,\ldots,k$. This implies that $I_1(\B) = J_1(\B) = K[x,y]$. According to Lemma \ref{lem8}, we obtain
 \begin{equation}\label{equ-maintheorem-4}
  \B \sim \begin{pmatrix} 1 & \\  & \M_1 \end{pmatrix},
 \end{equation}
 where $\M_1 \in K[x,y]^{(k-1)\times (k-1)}$. Let $\M' =  \begin{pmatrix} 1 & \\  & \M_1 \end{pmatrix}$. It follows from Lemma \ref{lem9} that
 \begin{equation*}\label{equ-maintheorem-5}
  J_i(\M_1) = J_{i+1}(\M') = J_{i+1}(\B) = K[x,y],
 \end{equation*}
 where $i = 1,\ldots,k-1$. Since $\det(\M) = (d_1(\M))^k \cdot \det(\M_0)$ and $\M_0 \sim \M'$, we have
 \begin{equation*}\label{equ-maintheorem-6}
  \det(\M_1) = \frac{1}{(d_1(\M))^k} \cdot (fy+g).
 \end{equation*}
 By the induction hypothesis, $\M_1$ is equivalent to its Smith normal form. That is,
 \begin{equation}\label{equ-maintheorem-7}
  \M_1 \sim \diag(\Psi_1, \ldots, \Psi_{k-1}),
 \end{equation}
 where $\Psi_1, \ldots, \Psi_{k-1}\in K[x,y]$ satisfy $\Psi_1\mid \Psi_2 \mid \cdots \mid \Psi_{k-1}$. Combining Equations \eqref{equ-maintheorem-1}--\eqref{equ-maintheorem-7}, we obtain
 \begin{equation*}\label{equ-maintheorem-8}
  \M \sim \diag(d_1(\M), \, d_1(\M)\cdot \Psi_1, \ldots, d_1(\M)\cdot \Psi_{k-1}).
 \end{equation*}
 Let $\Phi_1 = d_1(\M)$ and $\Phi_i = d_1(\M)\cdot \Psi_{i-1}$ for $i = 2, \ldots, k$. Then
 \begin{equation*}\label{equ-maintheorem-9}
  \M \sim \diag(\Phi_1, \Phi_2, \ldots, \Phi_k),
 \end{equation*}
 where $\Phi_i \mid \Phi_{i+1}$ for $i = 1, \ldots, k - 1$. It follows that $\M$ is equivalent to its Smith normal form.

 The proof is completed.
\end{proof}

 Theorem \ref{main-Theorem} provides an affirmative solution to Problem \ref{intro_problem}.

\begin{corollary}\label{main-corollary}
 Let $\M\in K[x,y]^{l\times m}$ with rank $r$, and $d_r(\M)= fy+g$, where $1\leq r \leq l$ and $f,g\in K[x]$ and $f \neq 0$. Then $\M$ is equivalent to its Smith normal form if and only if $J_i(\M) = K[x,y]$ for $i=1,\ldots,r$.
\end{corollary}

\begin{proof}
 The necessity is obvious from Lemma \ref{lem0}. It suffices to prove the sufficiency. Since $J_r(\M) = K[x,y]$, by Lemma \ref{lin-bose} there exist $\G_1\in K[x,y]^{l\times r}$ and $\F_1\in K[x,y]^{r\times m}$ such that
 \begin{equation}\label{equ-coro-1}
  \M = \G_1\F_1,
 \end{equation}
 where $\F_1$ is a ZLP matrix. By the Quillen-Suslin Theorem, there exists $\V\in \GL_m(K[x,y])$ such that
 \begin{equation}\label{equ-coro-2}
  \F_1\V = (\mathbf{I}_r, \mathbf{0}_{r\times(m-r)}).
 \end{equation}
 Combining Equations \eqref{equ-coro-1} and \eqref{equ-coro-2}, we have
 \begin{equation}\label{equ-coro-3}
  \M\V = (\G_1, \mathbf{0}_{l\times(m-r)}).
 \end{equation}
 By the fact that $\V$ is unimodular, it follows from Equation \eqref{equ-coro-3} and Lemma \ref{lem0} that $J_r(\G_1) = J_r(\M) = K[x,y]$. Using Lemma \ref{lin-bose} again, there exist $\G_2\in K[x,y]^{l\times r}$ and $\G_3\in K[x,y]^{r\times r}$ such that
 \begin{equation}\label{equ-coro-4}
  \G_1 = \G_2\G_3,
 \end{equation}
 where $\G_2$ is a ZRP matrix. According to the Quillen-Suslin Theorem, there exists $\U\in \GL_l(K[x,y])$ such that
 \begin{equation}\label{equ-coro-5}
  \U \G_2 = \begin{pmatrix} \mathbf{I}_r \\ \mathbf{0}_{(l-r)\times r}\end{pmatrix}.
 \end{equation}
 Combining Equations \eqref{equ-coro-3}--\eqref{equ-coro-5}, we obtain
 \begin{equation*}\label{equ-coro-6}
  \U \M \V = \begin{pmatrix} \G_3 & \mathbf{0}_{r\times (m-r)} \\ \mathbf{0}_{(l-r)\times r} & \mathbf{0}_{(l-r)\times (m-r)}\end{pmatrix}.
 \end{equation*}
 Therefore, by Lemma \ref{lem0} we get
 \[d_i(\G_3) = d_i(\M) ~ \text{and} ~ J_i(\G_3) = J_i(\M),\]
 where $i=1,\ldots,r$. According to Theorem \ref{main-Theorem}, there exist $\U_1,\V_1\in \GL_r(K[x,y])$ such that
 \begin{equation*}\label{equ-coro-7}
  \U_1 \G_3 \V_1 = \diag(\omega_1,\ldots,\omega_r),
 \end{equation*}
 where $\omega_i = \frac{d_i(\M)}{d_{i-1}(\M)}$ for $i=1,\ldots,r$. Let
 \[\U_2 = \begin{pmatrix} \U_1 &   \\   &  \mathbf{I}_{l-r} \end{pmatrix} ~ \text{and} ~
   \V_2 =  \begin{pmatrix} \V_1 &   \\   &  \mathbf{I}_{m-r} \end{pmatrix}.\]
 Then $\U_2,\V_2$ are unimodular matrices. Moreover, we have
 \begin{equation}\label{equ-coro-8}
  \U_2\U \M \V \V_2 = \begin{pmatrix} \diag(\omega_1,\ldots,\omega_r) & \mathbf{0}_{r\times (m-r)} \\ \mathbf{0}_{(l-r)\times r} & \mathbf{0}_{(l-r)\times (m-r)}\end{pmatrix}.
 \end{equation}
 It follows from Equation \eqref{equ-coro-8} that $\M$ is equivalent to its Smith normal form.

 The proof is completed.
\end{proof}

 Now we use an example to illustrate the effectiveness of Theorem \ref{main-Theorem}.

\begin{example}\label{equivalence-example}
 Let
 \[\M = \begin{pmatrix}
       x^2y^2-x^2y-xy^2+2x+y-1 & x^3y^2-x^3y-x^2y^2-x^2y+x^2+2xy-1 \\
        xy^2-xy-y^2+y+1 & x^2y^2-x^2y-xy^2+y
        \end{pmatrix}\]
 be a bivariate polynomial matrix in $\mathbb{Q}[x,y]^{2\times 2}$, where $\mathbb{Q}$ is the field of rational numbers, $y>x$ and $\prec$ is the lexicographic order.

 By calculation, we have $\det(\M) = (x-1)(xy-x-1)$. According to the definition of $I_1(\M)$, we have
 \[I_1(\M) = \langle \M[1,1],\M[1,2],\M[2,1],\M[2,2]\rangle,\]
 where $\M[i,j]$ is the $(i,j)$-th entry of $\M$, $1 \leq i,j \leq 2$. We compute a reduced Gr\"{o}bner basis \cite{Buchberger1995GB} of $I_1(\M)$ with respect to $\prec$ and obtain $\{1\}$. This implies that $J_1(\M) = I_1(\M) = \mathbb{Q}[x,y]$ and $d_1(\M) = 1$. Obviously, $J_2(\M) = \mathbb{Q}[x,y]$. Then we can apply Theorem \ref{main-Theorem} to reduce $\M$ to its equivalent Smith normal form.

 We use the Primitive Factorization Theorem to compute $\G,\F\in \mathbb{Q}[x,y]^{2\times 2}$ such that
 \begin{equation}\label{equ-exam-1}
  \M = \G\F ~ \text{and} ~ \det(\G) = x-1.
 \end{equation}
 According to the primitive factorization algorithm proposed in \cite{Guiver1982Poly}, we have
 \begin{equation*}\label{equ-exam-2}
  \G = \begin{pmatrix} 1 & 0 \\ 1 & x-1  \end{pmatrix} ~ \text{and} ~
  \F = \begin{pmatrix} \M[1,1] &  \M[1,2] \\  -(xy^2-xy-y^2+2) & -(xy-y-1)(xy-x-1) \end{pmatrix},
 \end{equation*}
 where $\det(\F) = xy-x-1$.  Set
 \[\P_1 = \begin{pmatrix} 1 & 0 \\ 1 & 1  \end{pmatrix} ~ \text{and} ~ \S_1 = \begin{pmatrix} 1 &   \\   & x-1 \end{pmatrix}.\]
 It is easy to see that
 \begin{equation}\label{equ-exam-3}
  \G = \P_1\S_1.
 \end{equation}
 Let $\F_1 = \F(x,\frac{x+1}{x})$. Then
 \[ \F_1 = \begin{pmatrix} 2x-1 & 0 \\ \frac{-2x^2+x+1}{x^2} & 0  \end{pmatrix}.\]
 Set $\F_2 = x^2 \cdot \F_1$. There exists $\P_2\in \GL_2(\mathbb{Q}[x])$ such that
 \begin{equation*}\label{equ-exam-4}
  \P_2 \F_2 = \begin{pmatrix} 1 & 0 \\ 0 & 0  \end{pmatrix},
 \end{equation*}
 where
 \[\P_2 = \begin{pmatrix} 2x-1 & 2x^2-x+1 \\ 2x^2-x-1 & 2x^3-x^2  \end{pmatrix}.\]
 Let
 \begin{equation}\label{equ-exam-5}
  \F_3 = \P_2\F.
 \end{equation}
 By Lemma \ref{lem1}, we obtain
 \begin{equation}\label{equ-exam-6}
  \F_3  = \S_2 \F_4,
 \end{equation}
 where $\S_2 = \begin{pmatrix} 1 &  \\  & xy-x-1  \end{pmatrix}$ and
 \[\F_4 = \begin{pmatrix} -xy^2+3xy-2x+y^2-y-1 & -x^2y^2+3x^2y-2x^2+xy^2-2x-y \\  -xy+2x+y-1 & -x^2y+2x^2+xy-1 \end{pmatrix}.\]
 Combining Equations \eqref{equ-exam-1}--\eqref{equ-exam-6}, we have
 \begin{equation*}\label{equ-exam-7}
  \M  = \P_1 \S_1 \P_2^{-1} \S_2 \F_4.
 \end{equation*}
 Let $\B =  \S_1 \P_2^{-1} \S_2$. Since $\P_1,\F_4$ are two unimodular matrices, we get
 \begin{equation}\label{equ-exam-8}
  \M  \sim \B.
 \end{equation}
 By calculation, we have
 \[\B =  \begin{pmatrix} 2x^3-x^2 & -(2x^2-x+1)(xy-x-1) \\ -2x^3+3x^2-1 & (2x^2-3x+1)(xy-x-1) \end{pmatrix}.\]
 Let $\P_3 = \begin{pmatrix} -4x^2+4x+1 & -4x^2-1 \\ 2x^3-3x^2+1 & 2x^3-x^2 \end{pmatrix}$. It is easy to compute that $\det(\P_3) = 1$. This implies that $\P_3\in \GL_2(\mathbb{Q}[x])$. Let $\F_5 = \P_3\B$. Then
 \[\F_5 =  \begin{pmatrix} 1 & -2xy+2x+2 \\ 0 & (x-1)(xy-x-1) \end{pmatrix}.\]
 Set $\P_4 = \begin{pmatrix} 1 & 2xy-2x-2 \\ 0 & 1 \end{pmatrix}$. It is easy to check that
 \begin{equation*}\label{equ-exam-9}
  \S = \F_5 \P_4 = \begin{pmatrix} 1 &  \\  & (x-1)(xy-x-1)  \end{pmatrix}.
 \end{equation*}
 Combining the above equations, we have
 \begin{equation}\label{equ-exam-10}
  \P_3\P_1^{-1}\M \F_4^{-1} \P_4 = \S.
 \end{equation}
 Let $\U = \P_3\P_1^{-1}$ and $\V = \F_4^{-1} \P_4$. Then
 \[\U = \begin{pmatrix} 4x+2 & -4x^2-1 \\ -2x^2+1 & 2x^3-x^2 \end{pmatrix}\]
 and
 \[\V = \begin{pmatrix} -x^2y+2x^2+xy-1 & -(x-1)(2xy-4x-y-2)(xy-x-1) \\ xy-2x-y+1 & 2x^2y^2-6x^2y+4x^2-3xy^2+5xy+y^2+y-3 \end{pmatrix}.\]
 It follows from Equation \eqref{equ-exam-10} that $\U\M\V = \S$.
\end{example}

\section{Concluding remarks}\label{sec5}

 In this paper, we have investigated the equivalence problem between bivariate polynomial matrices and their Smith normal forms. Let $\M\in K[x,y]^{l\times m}$ with rank $r$, where $1 \leq r \leq l$. Corollary \ref{main-corollary} shows that, for the case of $d_r(\M) = fy+g$ with $f,g\in K[x]$ and $f \neq 0$, Frost and Storey's assertion holds. Consequently, we have completely solved the equivalence problem between a class of bivariate polynomial matrices and their Smith normal forms. We hope that the methods proposed herein will facilitate research on the equivalence problem between multivariate polynomial matrices (with more than $2$ variables) and their Smith normal forms.

\section*{Acknowledgments}

 This research was supported by the Key Project of the National Natural Science Foundation of China under Grant No. 12494550\&12494551, the National Natural Science Foundation of China under Grant Nos. 12171469 and 12201210, the Sichuan Science and Technology Program under Grant No. 2024NSFSC0418, the Fundamental Research Funds for the Central Universities under Grant No. 2682024ZTPY052, and the STU Scientific Research Initiation under Grant No. NTF24023T.

\bibliographystyle{elsarticle-num-names}
\bibliography{Biv_ref}

\end{document}